\documentclass[twoside]{article}

\usepackage[round]{natbib}
\usepackage[accepted]{aistats2019}

\usepackage{amsmath,amssymb,amsfonts,graphicx}
\usepackage{latexsym}
\usepackage{amsthm}
\usepackage{url}
\usepackage{algorithm}
\usepackage{algpseudocode}

\newcommand{\database}{{\mathcal{D}}}    % database
\newcommand{\noise}{{X}}    %  symbol denoting the noise
\newcommand{\KM}{{\mathcal{K}}}    % randomized mechanism
\newcommand{\e}{{\epsilon}}    % epsilon
\newcommand{\loss}{{\mathcal{L}}}    % cost function / loss function
\newcommand{\D}{{\Delta}}    % query sensitivity
\newcommand{\p}{{\mathcal{P}}}  % probability measure
\newcommand{\sP}{{\mathcal{SP}}}  % set of probability measures satisfying differential privacy constraint
\newcommand{\psym}{{\p_{\text{sym} }}}
\newcommand{\pSymMon}{{\p_{\text{sym, mon} }}}

\newcommand{\R}{{\mathbb{R}}}  % real number
\newtheorem{theorem}{Theorem}
\newtheorem{lemma}{Lemma}
\newtheorem{definition}{Definition}
\newtheorem{corollary}[theorem]{Corollary}
\newtheorem{property}{Property}

\begin{document}

\twocolumn[

\aistatstitle{Optimal Noise-Adding Mechanism in Additive Differential Privacy}

\aistatsauthor{Quan Geng \And Wei Ding \And  Ruiqi Guo  \And  Sanjiv Kumar}
\aistatsaddress{Google AI \And  Google AI \And Google AI \And Google AI} ]

\begin{abstract}
We derive the optimal $(0, \delta)$-differentially private query-output independent noise-adding mechanism for single real-valued query function under a general cost-minimization framework. Under a mild technical condition, we show that the optimal noise probability distribution is a uniform distribution with a probability mass at the origin. We explicitly derive the optimal noise distribution for general $\ell^p$ cost functions, including $\ell^1$ (for noise magnitude) and $\ell^2$ (for noise power) cost functions, and show that the probability concentration on the origin occurs when $\delta > \frac{p}{p+1}$. Our result demonstrates an improvement over the existing Gaussian mechanisms  by a factor of two and three for $(0,\delta)$-differential privacy in the high privacy regime in the context of minimizing the noise magnitude and noise power, and the gain is more pronounced in the low privacy regime. Our result is consistent with the existing result for $(0,\delta)$-differential privacy in the discrete setting, and identifies a probability concentration phenomenon in the continuous setting.
\end{abstract}

\section{Introduction} \label{sec:intro}
Differential privacy, introduced by \citet{DMNS06}, is a framework to quantify to what extent individual privacy in a statistical dataset is preserved while releasing useful aggregate information about the dataset.
Differential privacy provides strong privacy guarantees by requiring the near-indistinguishability of whether an individual is in the dataset or not based on the released information. 
For more motivation and background of differential privacy, we refer the readers to the survey by \citet{DPsurvey} and the book by \citet{DPbook}.

The classic differential privacy is called $\e$-differential privacy, which imposes an upper bound $e^\e$ on the multiplicative distance of the probability distributions of the randomized query outputs for any two neighboring datasets, and the standard approach for preserving $\e$-differential privacy is to add a Laplacian noise to the query output.
Since its introduction, differential privacy has spawned a large body of research in differentially private data-releasing mechanism design, and the noise-adding mechanism has been applied in many machine learning algorithms to preserve differential privacy, e.g., 
logistic regression \citep{CM08}, 
empirical risk minimization \citep{ERM}, 
online learning \citep{Jain12}, 
statistical risk minimization \citep{Duchi12}, 
deep learning \citep{Shokri15, Abadi2016, Phan2016, Agarwal18}, 
hypothesis testing \citep{HT18},
matrix completion \citep{JainMC},
expectation maximization \citep{EM},
and principal component analysis \citep{PCA, PCA2}.

To fully make use of the randomized query outputs, it is important to understand the fundamental trade-off between privacy and utility (accuracy).  \citet{Ghosh09} studied a very general utility-maximization framework for a single count query with sensitivity one  under  $\epsilon$-differential privacy. 
\cite{minimax10} derived the optimal noise probability distributions for a single count query with sensitivity one for minimax (risk-averse) users.
\citet{GV_IT_Epsilon} derived the optimal $\e$-differentially private noise adding mechanism for single real-valued query function with arbitrary query sensitivity, and show that the optimal noise distribution has a staircase-shaped probability density function. \cite{GV_2_Dimension} generalized the result in \cite{GV_IT_Epsilon} to two-dimensional query output space for the $\ell^1$ cost function, and show the optimality of a two-dimensional staircase-shaped probability density function. \cite{DomingoFerrer2013} also independently derived the staircase-shaped noise probability distribution under a different optimization framework.

A relaxed notion of $\e$-differential privacy is $(\e,\delta)$-differential privacy, introduced by \cite{DKMMN06}. The common interpretation of $(\e,\delta)$-differential privacy is that it is $\e$-differential privacy ``except with probability $\delta$'' \citep{RenyiDP}.
The standard approach for preserving $(\e,\delta)$-differential privacy is the Gaussian mechanism, which adds a Gaussian noise to the query output.
\cite{GV_IT_Approximate} studied the trade-off between utility and privacy for a single \emph{integer-valued} query function in $(\epsilon, \delta)$-differential privacy. \cite{GV_IT_Approximate} show that for $\ell^1$ and $\ell^2$ cost functions, the discrete uniform noise distribution is optimal for $(0,\delta)$-differential privacy when the query sensitivity is one, and is \emph{asymptotically} optimal as $\delta \to 0$ for arbitrary query sensitivity. \cite{GW2_18} extend the result for single real-valued query functions under $(\e, \delta)$-differential privacy and show that the truncated Laplacian mechanism is asymptotically optimal in various high privacy regimes.
\cite{icmlGaussian} improved the classic analysis of the Gaussian mechanism for $(\e,\delta)$-differential in the high privacy regime ($\e \to 0$), and develops an optimal Gaussian mechanism whose variance is calibrated directly using the Gaussian cumulative density function instead of a tail bound approximation.

$(\e, 0)$-differential privacy and $(0,\delta)$-differential privacy can be viewed as two special cases of the commonly used $(\epsilon, \delta)$-differential privacy paradigm. While $(\epsilon, 0)$-differential privacy is well studied and exact optimality result has been obtained, little is known about $(0, \delta)$-differential privacy. Characterizing the privacy-utility tradeoff in $(0, \delta$)-differential privacy is important towards understanding the fundamental privacy and utility tradeoff in $(\epsilon, \delta)$-differential privacy.

\subsection{Our Contributions}

In this work,  we characterize the fundamental trade-off between privacy and utility in $(0, \delta)$-differential privacy for a single read-valued query function. Within the class of query-output independently noise-adding mechanisms, we derive the optimal noise distribution for $(0, \delta)$-differential privacy under a general cost-minimization framework similar to \cite{Ghosh09,minimax10,GV_ISIT14, GV_2_Dimension, GV_IT_Approximate}.
Under a mild technical condition on the noise probability distribution\footnote{In this work, we assume that the noise probability distribution has higher probability over the small noise than the big noise. This condition is satisfied by virtually all probability distributions used in differential privacy, including the uniform distribution, the Laplacian distribution, the truncated Laplacian distribution, the Gaussian distribution, and the staircase distribution. While the optimality result in this paper depends on this assumption, we believe this assumption can be done away.}, we show that the optimal noise probability distribution is a uniform  distribution with a probability mass at the origin, which can be viewed as the distribution of the product of a uniform random variable and a Bernoulli random variable. The probability mass on the origin can be zero or non-zero, depending on the value of $\delta$.
We explicitly derive the optimal noise distribution for general $\ell^p$ cost functions, including $\ell^1$ (for noise magnitude) and $\ell^2$ (for noise power) cost functions, and show that the probability concentration on the origin occurs when $\delta > \frac{p}{p+1}$. 

Compared with the improved Gaussian mechanisms for $(0,\delta)$-differential privacy \citep{icmlGaussian}, our result demonstrates a two-fold and three-fold improvement in the high privacy regime in the context of minimizing the noise magnitude and noise power, respectively. The improvement is more pronounced in the low privacy regime.

Comparing the exact optimality results of $\e$-differential privacy and $(0,\delta)$-differential privacy, we show that given the same amount of privacy constraint, $(0,\delta)$-differential privacy yields a higher utility than $\e$-differential privacy in the high privacy regime.

Our result is consistent with the existing result for $(0,\delta)$-differential privacy in the discrete setting \citep{GV_IT_Approximate} which shows that the discrete uniform distribution is optimal for an \emph{integer-valued} query function when the query sensitivity is one, and asymptotically optimal as $\delta \to 0$ for general query sensitivity. Interestingly, our result identifies a probability concentration phenomenon in the continuous setting for single \emph{real-valued} query function.

\subsection{Organization}

The paper is organized as follows. In Section \ref{sec:formulation}, we give some preliminaries on differential privacy, and formulate the trade-off between privacy and utility under $(0, \delta)$-differential privacy for a single real-valued query function as a functional optimization problem. 
Section \ref{sec:result} presents the optimal noise probability distribution preserving $(0, \delta)$-differential privacy, subject to a mild technical condition. Section \ref{sec:application} applies our main result to a class of momentum cost functions, and derives the explicit forms of the optimal noise probability distributions with minimum noise magnitude and noise power, respectively. Section \ref{sec:comparison} compares our result with the improved Gaussian mechanism in the context of minimizing noise magnitude and noise power. 
%Section \ref{sec:conclusion} concludes this paper.

\section{Problem Formulation} \label{sec:formulation}
 In this section, we first give some preliminaries on differential privacy, and then formulate the trade-off between privacy and utility under $(0, \delta)$-differential privacy for a single real-valued query function as a functional optimization problem. 

\subsection{Background on Differential Privacy}

Consider a real-valued query function
\begin{align*}
    q: \database \rightarrow \R,
\end{align*}
where $\database$ is the set of all possible datasets. The real-valued query function $q$ will be applied to a dataset, and the query output is a real number. Two datasets $D_1, D_2 \in \database$ are called neighboring datasets if they differ in at most one element, i.e.,  one is a proper subset of the other and the larger dataset contains just one additional element \citep{DPsurvey}. A randomized query-answering mechanism $\KM$ for the query function $q$ will randomly output a number with probability distribution depending on query output $q(D)$, where $D$ is the dataset.

\begin{definition}[$\e$-differential privacy \citep{DPsurvey}]
    A randomized mechanism $\KM$ gives $\e$-differential privacy if for all data sets $D_1$ and $D_2$ differing on at most one element, and any measurable set $S \subset \text{Range}(\KM)$,
    \begin{align}
        \text{Pr}[\KM(D_1) \in S] \le e^\e \;  \text{Pr}[\KM(D_2) \in S], \label{eqn:dpgeneral}
     \end{align}
     where $\KM(D)$ is the random output of the  mechanism $\KM$ when the query function $q$ is applied to the dataset $D$.
\end{definition}

The differential privacy constraint \eqref{eqn:dpgeneral} imposes an upper bound $e^\e$ on the multiplicative distance of the two probability distributions. It
essentially requires that for all neighboring datasets, the probability distributions of the output of the randomized mechanism should be approximately the same. Therefore, for any individual record, its presence or absence in the dataset will not significantly affect the output of the mechanism, which makes it hard for adversaries with arbitrary background knowledge to make inference on any individual from the released query output information. The parameter $\e \in (0, +\infty)$ quantifies how private the mechanism is: the smaller $\e$ is, the more private the randomized mechanism is.

The standard approach to preserving $\e$-differential privacy is to perturb the query output by adding a random noise with Laplacian distribution proportional to the sensitivity $\Delta$ of the query function $q$, where the sensitivity of a real-valued query function is defined as:

\begin{definition}[Query Sensitivity \citep{DPsurvey}]
For a real-valued query function $q: \database \rightarrow \R$, the sensitivity of $q$ is defined as
\begin{align*}
    \D := \max_{D_1,D_2 \in \database} | q(D_1) - q(D_2)|,
\end{align*}
for all $D_1,D_2$ differing in at most one element.
\end{definition}

Introduced by \cite{DKMMN06}, a relaxed version of $\epsilon$-differential privacy is $(\epsilon, \delta)$-differential privacy, which relaxes the constraint \eqref{eqn:dpgeneral} with an additive term $\delta \in [0, 1]$.

\begin{definition}[$(\e,\delta)$-differential privacy \citep{DKMMN06}]
A randomized mechanism $\KM$ gives $(\e, \delta)$-differential privacy if for all data sets $D_1$ and $D_2$ differing on at most one element, and all $S \subset \text{Range}(\KM)$,
\begin{align*}
    \text{Pr}[\KM(D_1) \in S] \le e^\e \;  \text{Pr}[\KM(D_2) \in S] + \delta.
 \end{align*}
\end{definition}

In the special case where $\e = 0$, the constraint for $(0, \delta)$-differential privacy is
\begin{align}
    \text{Pr}[\KM(D_1) \in S] \le \text{Pr}[\KM(D_2) \in S] + \delta. \label{eq:dp_zero_delta_constraint}
\end{align}

It is ready to see that $(0, \delta)$-differential privacy puts an upper bound $\delta$ on the \emph{additive} distance of the two probability distributions.

\subsection{$(0, \delta)$-Differential Privacy Constraint on the Noise Probability Distribution}

A standard approach for preserving differential privacy is query-output independent noise-adding mechanisms, where a random noise is added to the query output.  Given a dataset $D$, a query-output independent noise-adding mechanism $\KM$ will release the query output $t = q(D)$ corrupted by an additive random noise $X$ with probability distribution $\p$:
\begin{align}
    \KM(D) = t + X.
\end{align}

The $(0, \delta)$-differential privacy constraint \eqref{eq:dp_zero_delta_constraint}  on $\KM$ is that for any $t_1,t_2 \in \R$ such that $|t_1 - t_2| \le \D $ (corresponding to the query outputs for two neighboring datasets),
\begin{align*}
    \p (S) \le \p(S + t_1 - t_2) + \delta, \forall \; \text{measurable set} \; S \subset \R, 
\end{align*}
where $ \forall t \in \R$,  $S+t$ is defined as the set  $\{s+t \, | \, s \in S\}$.

Equivalently, the $(0,\delta)$-differential privacy constraint on the noise probability distribution $\p$ is
\begin{align}
    \p (S) \le \p(S + d) + \delta, \forall \; |d| \le \D, \text{measurable set} \; S \subset \R.  \label{eqn:diffgeneralnoise1}
\end{align}

\subsection{Utility Model}

Consider a cost function $\loss(\cdot): \R \to \R$, which is a function of the additive noise in the query-output noise-adding mechanism. Given an additive noise $x$, the cost is $\loss(x)$, and thus the expectation of the cost over $\p$ is
\begin{align}
    \int_{x\in \R} \loss(x) \p(dx).\label{eqn:obj1}
\end{align}

Our objective is to minimize the expectation of the cost over the noise probability distribution for preserving $(0,\delta)$-differential privacy.

\subsection{Optimization Problem}

Combining the differential privacy constraint \eqref{eqn:diffgeneralnoise1} and the objective function \eqref{eqn:obj1}, we formulate a functional optimization problem:
\begin{align}
    \mathop{\text{minimize}}\limits_{\p} & \ \int_{x \in \R} \loss(x) \p(dx) \label{eqn:objecinopt}\\
    \text{subject to} \; & \forall \ \text{measurable set} \ S\subseteq \R, \ \forall |d| \le \D.  \nonumber \\
     & \ |\p (S) - \p(S + d)| \le  \delta \label{eq:mainconstraint}
\end{align}

\section{Main Result} \label{sec:result}
In this section, we solve the functional optimization problem \eqref{eqn:objecinopt} for $(0, \delta)$-differential privacy, and present our main result in Theorem \ref{thm:main}. Under a mild technical condition on the probability distribution (see Property \ref{property:noise-decreasing}), we show that the optimal noise probability distribution is a uniform  distribution with a probability mass at the origin, which can be viewed as the distribution of the product of a uniform random variable and a Bernoulli random variable.

We assume that the cost function $\loss(\cdot)$ satisfies a natural property.

\begin{property} \label{property1}
    $\loss(x)$ is a symmetric function, and monotonically increasing for $x \ge 0$, i.e, $\loss(x)$ satisfies
    \begin{align*}
        \loss(x) &= \loss(-x), \forall \; x \in \R,
    \end{align*}
    and
    \begin{align*}
        \loss(x) &\le \loss(y), \forall \; 0 \le x \le y.
    \end{align*}
\end{property}

First, we show that without loss of generality, we only need to consider symmetric noise probability distributions.

\begin{lemma}\label{lem:symmetry}
Given a noise probability distribution $\p$ satisfying \eqref{eq:mainconstraint}, there exists a probability distribution $\hat{\p}$ such that $\hat{\p}$ satisfies \eqref{eq:mainconstraint}, and 
\begin{align*}
    \hat{\p}(S) = \hat{\p}(-S), \forall \; \text{measurable set}\ S \subseteq \R,  
\end{align*}
and     
\begin{align*}
    \int_{x\in \R} \loss(x) \hat{\p}(dx) = \int_{x\in \R} \loss(x) \p(dx). 
\end{align*}
\end{lemma}

\begin{proof}
Define $\hat{\p}$ as follows: $ \forall \; \text{measurable set} S \subseteq \R$,
\begin{align*}
    \hat{\p}(S) := \frac{\p(S) + \p(-S)}{2}.
\end{align*}

It is ready to see $\hat{\p}$ is a symmetric probability distribution. As the loss function $\loss(\cdot)$ is symmetric, we have $\int_{x\in \R} \loss(x) \hat{\p}(dx) = \int_{x\in \R} \loss(x) \p(dx).$

Next we show that $\hat{\p}$ also satisfies the differential privacy constraint. Indeed, $\forall \ measurable\ set\ S \subseteq \R$ and $d$ such that $|d| \le \D$, we have
\begin{align*}
 &\;\;\; \;\;|\hat{\p} (S) - \hat{\p}(S + d)| \\
 & = | \frac{\p(S) + \p(-S)}{2} -  \frac{\p(S+d) + \p(-S-d)}{2} |  \\
 & \le \frac{|\p(S) - \p(S+d)|}{2} +  \frac{|\p(-S)+ \p(-S-d)|}{2} \\
 & \le \frac{\delta}{2} + \frac{\delta}{2} \\
 & = \delta.
\end{align*}
\end{proof}

Due to Lemma \ref{lem:symmetry}, we can restrict ourselves to symmetric noise probability distributions. 

As the loss function $\loss(\cdot)$ is monotonically increasing as the noise becomes bigger, we impose a mild and natural condition on the symmetric noise probability contribution, which requires the noise probability distribution to have bigger probability measure on the small noise than the large noise. More precisely,
\begin{property} \label{property:noise-decreasing}
    Given a symmetric probability measure $\p \in \psym$, $\p$ is monotonically decreasing if
    \begin{align*}
        \p(S) \ge \p(S + a), \forall a \ge 0, \ \text{measurable set} \ S \subseteq [0, +\infty)
    \end{align*}
\end{property}

Property \ref{property:noise-decreasing} is satisfied by a large class of probability distributions, including the uniform distribution, the Laplacian distribution and the Gaussian distribution.

Let $\pSymMon$ denote the set of symmetric probability measures which are monotonically decreasing. 

\begin{lemma}\label{lem:no-singular-mass}
Given $\p \in \pSymMon$, then for any $t \neq 0$, $\p(\{t\}) = 0$, i.e., $\p$ cannot have a non-zero probability mass on any singular point except the origin $t=0$. 
\end{lemma}

\begin{proof}
Suppose there exists $t \neq 0$ such that $\p(\{t\}) \neq 0 $. Since $\p$ is symmetric, we can assume $t >0$. Since $\p$ is monotonically decreasing in $[0, +\infty)$, for any $t' \in (0, t), \p(\{t'\}) \geq \p(\{t\})$, which implies $\p((0,t)) = +\infty$. This contradicts with the fact that $\p$ is a probability measure.  
\end{proof}

Within the classes of monotonically decreasing probability distributions, we identify a sufficient and necessary condition for preserving $(0,\delta)$-differential privacy.

\begin{theorem}\label{thm:sufficient-necessary}
Given $\p \in \pSymMon$, $\p$ satisfies the $(0, \delta)$-differential privacy constraint \eqref{eq:mainconstraint}, if and only if 
\begin{align*}
\p([-\frac{\D}{2}, \frac{\D}{2}]) \le \delta.
\end{align*}
\end{theorem}
\begin{proof}
First we show that it is a necessary condition. Assume $\p$ satisfies the $(0, \delta)$-differential privacy constraint \eqref{eq:mainconstraint}. Consider $S = [-\frac{\D}{2}, +\infty)$ and $d= \D$ in \eqref{eq:mainconstraint}, and we have
\begin{align*}
|\p([-\frac{\D}{2}, +\infty)) - \p([\frac{\D}{2}, +\infty)) | \le \delta,
\end{align*}
and thus $\p([-\frac{\D}{2}, \frac{\D}{2})) \le \delta$. Due to Lemma \ref{lem:no-singular-mass}, $\p([-\frac{\D}{2}, \frac{\D}{2}]) = \p([-\frac{\D}{2}, \frac{\D}{2})) \le \delta$.

Next we show that it is a sufficient condition. Assume $\p([-\frac{\D}{2}, \frac{\D}{2}]) \le \delta$. As $\p$ is symmetric and the differential privacy constraint \eqref{eq:mainconstraint} applies to all measurable subset $S \subseteq \R$ and all $d$ such that $|d| \le \D$, it is equivalent to show that $\p(S) - \p(S+d) \le \delta, \forall d \in (0, \D]$.

Since $\p$ is symmetric and monotonically decreasing in $[0, +\infty)$, $\p(S) - \p(S+d)$ is maximized when $S = [-\frac{d}{2}, +\infty)$. Therefore, $\forall d \in (0, \D]$,
\begin{align*}\
&\;\;\;\; \p(S) - \p(S+d) \\
&\le \p([-\frac{d}{2}, +\infty)) - \p([-\frac{d}{2}, +\infty) + d) \\
&= \p([-\frac{d}{2}, \frac{d}{2})) \\
&\le \p([-\frac{\D}{2}, \frac{\D}{2})) \\
&\le \delta.
\end{align*}

This concludes the proof of Theorem \ref{thm:sufficient-necessary}.
\end{proof}

Consider a class of probability distributions $\{P_\alpha\}$ parameterized by $\alpha \in [0, \delta)$, where $\p_\alpha$ is defined as
\begin{align*}
\p_\alpha(\{0\}) = \alpha
\end{align*}
and except the point $t=0$, $\p_\alpha$ has a uniform probability distribution over the set $[-\frac{1-\alpha}{\delta-\alpha}\frac{\D}{2}, \frac{1-\alpha}{\delta-\alpha}\frac{\D}{2}] \setminus \{0\}$
with probability density $\frac{\delta - \alpha}{\D}$ (see Figure \ref{fig:p_alpha}).

\begin{figure}[H]
\centering
\includegraphics[width=0.40\textwidth]{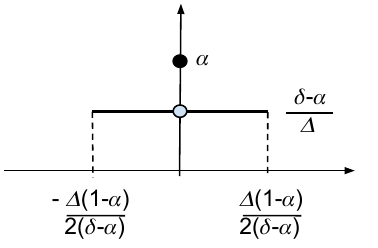}
\caption{Probability distribution of $\p_\alpha$. $\p_\alpha$ has a probability mass $\alpha \in [0, \delta)$ at the origin, and has a uniform distribution over $[-\frac{1-\alpha}{\delta-\alpha}\frac{\D}{2}, \frac{1-\alpha}{\delta-\alpha}\frac{\D}{2}] \setminus \{0\}$ with probability density $\frac{\delta - \alpha}{\D}$.}
\label{fig:p_alpha}
\end{figure}

Let $\sP$ denote the set of all symmetric and monotonically decreasing probability distributions satisfying the $(0,\delta)$-differential privacy \eqref{eq:mainconstraint}.
Our main result on the optimal noise probability distribution is:
\begin{theorem}\label{thm:main}
If the cost function $\loss(x)$ satisfies Property \ref{property1}, then for any $\D >0$ and $0 < \delta < 1$,
\begin{align*}
    \inf_{\p \in \sP}  \int_{x \in \R} \loss (x)  \p(dx)  = \inf_{\alpha \in [0, \delta)}  \int_{x \in \R} \loss (x)  \p_\alpha(dx).
\end{align*}
\end{theorem}

\begin{proof}
First note that for any $\alpha \in [0, \delta)$, $\p_\alpha$ is symmetric and monotonically decreasing in $[0, +\infty)$, and
\begin{align*}
 \p_\alpha([-\frac{\D}{2}, \frac{\D}{2}]) = \alpha + \frac{\delta - \alpha}{\D} \D = \delta.  
\end{align*}

Therefore, due to Theorem \ref{thm:sufficient-necessary}, $\p_\alpha$ satisfies the $(0, \delta)$-differential privacy constraint \eqref{eq:mainconstraint}, and thus $\p_\alpha \in \sP$.

Applying a similar argument as in Lemma 20 of \cite{GV_IT_Epsilon}, we can use a sequence of symmetric and piece-wise linear probability density function with probability mass concentration in the origin to approximate any $\p \in \sP$ (see Figure \ref{fig:discrete}). More precisely, given a probability distribution $\p \in \sP$ which may have non-zero probability mass at $x=0$, for positive integer $i \in N$, define the probability distribution $\p_i$  as follows:
\begin{align*}
\p_i(\{0\}) := \p(\{0\})
\end{align*}
and over the set $\R \setminus \{0\}$, $\p_i$ has a symmetric probability density function $f_i(x)$ with
\begin{align*}
    f_i(x) = \begin{cases}
        a_k \triangleq \frac{ \p((k\frac{\D}{2i}, (k+1)\frac{\D}{2i}])}{\frac{\D}{2i}}  & x\in (k\frac{\D}{2i}, (k+1)\frac{\D}{2i} ], \\ %k \in \N \\
        f_i(-x) &   x < 0
    \end{cases}     
\end{align*}
It is easy to see that $\p_i([-\frac{\D}{2}, \frac{\D}{2}]) = \p([-\frac{\D}{2}, \frac{\D}{2}]) \le \delta$, and thus due to Theorem \ref{thm:sufficient-necessary}, $\p_i \in \sP$. Due to the definition of Riemann-Stieltjes integral, we have
\begin{align*}
    \lim_{i \to +\infty}   \int_{x \in \R} \loss (x)  \p_i(dx)   =  \int_{x \in \R} \loss (x)  \p(dx)
\end{align*}

\begin{figure}[H]
\centering
\includegraphics[width=0.45\textwidth]{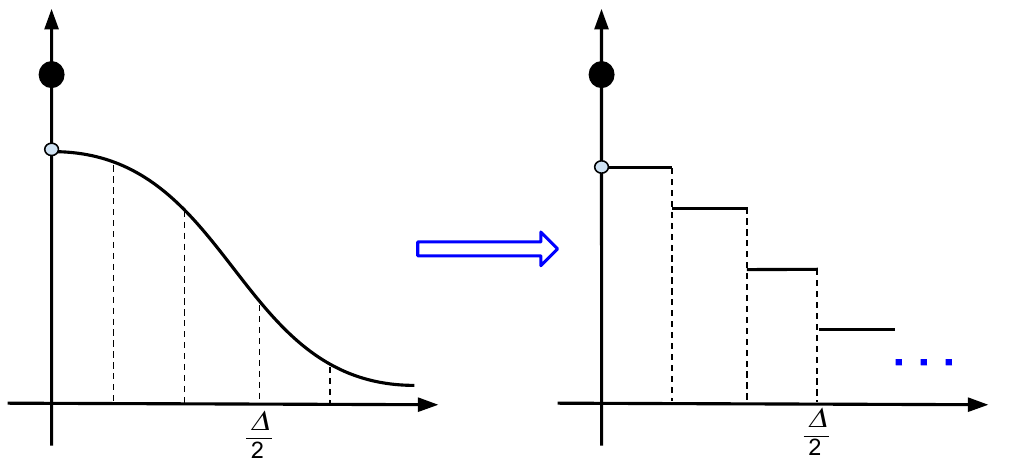}
\caption{Discretize a probability distribution with a piecewise-constant probability density function and a probability mass at the origin.}
\label{fig:discrete}
\end{figure}

Therefore, we only need to consider probability distributions with a probability mass on the origin and a symmetric monotonically decreasing piecewise-constant probability density function on $\R \setminus \{0\}$.

First we show that without loss of generality, we can assume the probability density function in $[\frac{\D}{2}, +\infty)$ is a step function, i.e., there exists a $t^* > \frac{\D}{2}$ such that probability density function is a constant in $[\frac{\D}{2}, t^*]$ and is zero in $(t^*, +\infty)$. Indeed, we can re-arrange the probability distribution in $[\frac{\D}{2}, +\infty)$ to make the probability density function to be uniform within certain interval $[\frac{\D}{2}, t^*]$ with the probability density the same as the previous bucket (see Figure \ref{fig:t_star}). This will not increase the cost, due to the fact that $\loss(\cdot)$ is a monotonically increasing function on $[0, +\infty)$. Since we are not changing the probability distribution in $[-\frac{\D}{2}, \frac{\D}{2}]$, due to Theorem \ref{thm:sufficient-necessary}, the probability distribution after the re-arrangement also satisfies the differential privacy constraint \eqref{eq:mainconstraint}.  

\begin{figure}[H]
\centering
\includegraphics[width=0.40\textwidth]{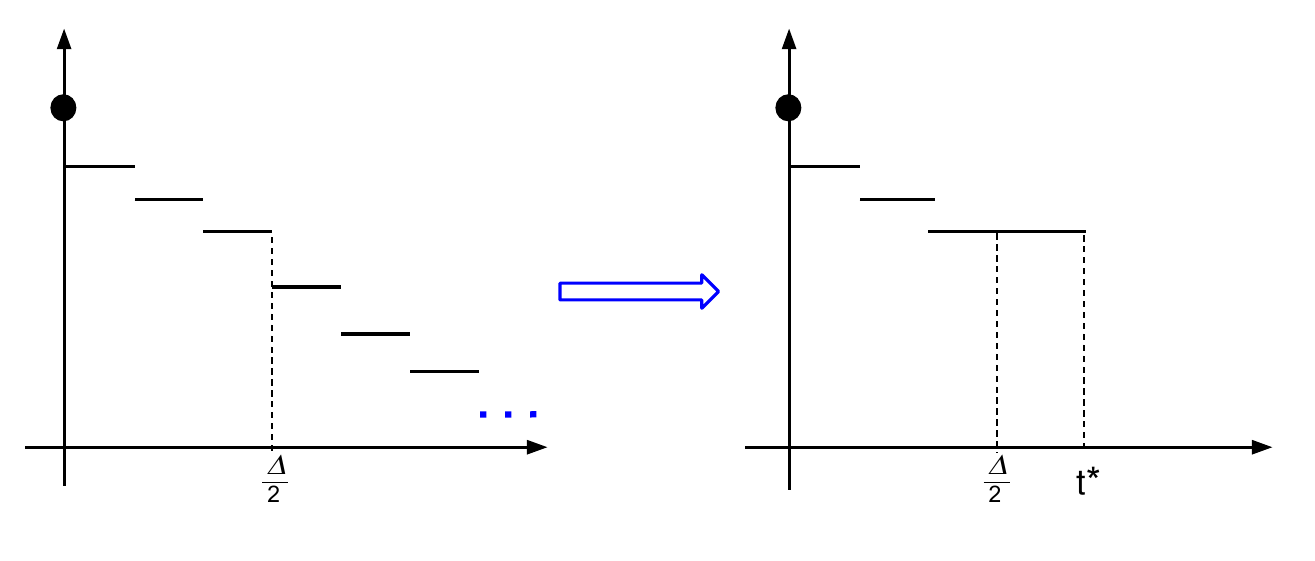}
\caption{Re-arrange the probability distribution in $[\frac{\D}{2}, +\infty)$ to be a step.}
\label{fig:t_star}
\end{figure}

Then we show that the probability distribution in $(0, \frac{\D}{2})$ shall be uniform as well. Indeed, if the distribution in $(0, \frac{\D}{2})$ is not uniform, we can decrease the probability density over $(0,\frac{\D}{2})$ to be the same as the point $\frac{\D}{2}$, and move the extra probability mass to the origin point (see Figure \ref{fig:zero_mass}). Due to the fact that $\loss(\cdot)$ is a monotonically increasing function on $[0, +\infty)$, this will not increase the cost. As $\p([-\frac{\D}{2}, \frac{\D}{2}])$ is unchanged, the new probability distribution satisfies the differential privacy constraint \eqref{eq:mainconstraint}.

\begin{figure}[H]
\centering
\includegraphics[width=0.40\textwidth]{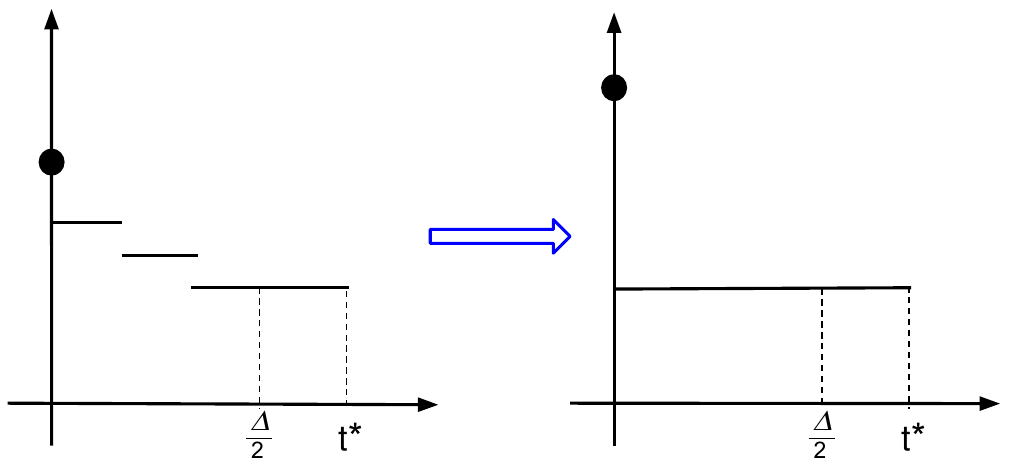}
\caption{Re-arrange the probability distribution in $(0, \frac{\D}{2})$ to be uniform and put the extra probability mass at the origin. }
\label{fig:zero_mass}
\end{figure}

In the last step, we show that $\p([-\frac{\D}{2}, \frac{\D}{2}])$ should be exactly $\delta$. If it is strictly less than $\delta$, then we can reduce the probability density over $(0, t^*)$ and increase $\alpha$ to make $\p([-\frac{\D}{2}, \frac{\D}{2}]) = \delta$. Similarly, due to the property of $\loss(\cdot)$ and Theorem \ref{thm:sufficient-necessary}, we conclude that this reduces the cost while preserving the differential privacy constraint.

This concludes the proof of Theorem \ref{thm:main}.
\end{proof}

A natural and simple algorithm to generate random noise with probability distribution $\p_\alpha$ is given in Algorithm \ref{algo:uniform_mech}.

%%%%%%%%%%%%%%%%%%%%%%    include the algorithm     %%%%%%%%%%%%%%%%%%%%%%%%%%%%%
\begin{algorithm}
\caption{Generation of a Random Variable of Uniform Distribution with a Probability Concentration on the Origin}
\label{algo:uniform_mech}
\begin{algorithmic}
\State \textbf{Input: } $\delta$, $\D$, and  $\alpha \in [0,\delta)$.
\State \textbf{Output: } $\noise$, a random variable of uniform distribution with a probability concentration on the origin, paremeterized by $\delta, \D$ and $\alpha$.
\\
\State Generate a binary random variable $B$ with 
\begin{align*}
\text{Pr}[B = 0] &= \alpha, \\
\text{Pr}[B = 1] & = 1 - \alpha.
\end{align*}
\State Generate a random variable $U$ uniformly distributed in $[-\frac{1-\alpha}{\delta-\alpha}\frac{\D}{2}, \frac{1-\alpha}{\delta-\alpha}\frac{\D}{2}]$.
\State $\noise \gets B \cdot U$.
\State Output $\noise$.
\end{algorithmic}
\end{algorithm}

\section{Applications} \label{sec:application}
In this section, we apply our main result Theorem \ref{thm:main} to derive an explicit expression for the parameter $\alpha$ in the optimal noise probability distribution $\p_\alpha$ for the class of $\ell^p$ momentum cost functions in Theorem \ref{thm:noise-general-n}. Applying Theorem \ref{thm:noise-general-n} to the cases $p=1$ and $p=2$, we get the optimal noise probability distribution with minimum noise amplitude and minimum noise power for $(0,\delta)$-differential privacy, respectively.

Let $V(\p) := \int_{x \in \R} \loss (x)  \p(dx)$, i.e., $V(\p)$ denote the expectation of the cost given the noise probability distribution $\p$ for the cost function $\loss(\cdot)$.

% theorem for general l^n cost function
\begin{theorem}\label{thm:noise-general-n}
Given $0 < \delta <1$ and the query sensitivity $\D > 0$. For the general momentum cost function $\ell^p(x) := |x|^p$, where $p > 0$, the optimal noise probability distribution to preserve $(0, \delta)$-differential privacy with query sensitivity $\D$ is $\p_\alpha*$ with
\begin{align*}
\alpha^*=
    \begin{cases}
      0,           & \text{for } \delta \in (0, \frac{p}{p+1}] \\
      (p+1)\delta - p, & \text{for } \delta \in (\frac{p}{p+1}, 1)
    \end{cases}
\end{align*}
and the minimum cost is
\begin{align*}
V(\p_{\alpha^*}) =  
    \begin{cases}
      \frac{\D^p}{2^p (p+1)\delta^p}, & \text{for } \delta \in (0, \frac{p}{p+1}] \\
      \frac{(p+1)^{p}}{2^p p^p} (1-\delta) \D^p,  & \text{for }  \delta \in (\frac{p}{p+1}, 1)
    \end{cases}
\end{align*}

\end{theorem}

% proof for general l^p cost function
\begin{proof}
It is easy to see that the momentum cost function satisfies Property \ref{property1}. We can compute the cost $V(\p_\alpha)$ via
\begin{align*}
  V(\p_\alpha) &= \int_{x \in \R} |x|^p \p_\alpha(dx) \\
  &= 2 \int_{0}^{\frac{1-\alpha}{\delta-\alpha}\frac{\D}{2}} x^p \frac{\delta - \alpha}{\D}dx \\
  &= \frac{\D^p}{(p+1)2^p}\frac{(1-\alpha)^{p+1}}{(\delta-\alpha)^p}.
\end{align*}

Define $f(\alpha) := \frac{(1-\alpha)^{p+1}}{(\delta-\alpha)^p}$.
As $f(\alpha)$ is a continuous function of $\alpha$ over $[0, \delta)$, and $f(0) = \frac{1}{\delta^p}$, and $f(\delta) = \infty$, the minimum is achieved at either $\alpha = 0$ or the point where the derivative $f'(\alpha) = 0$.
 
Compute the derivative of $f(\alpha)$ via
\begin{align*}
  &\;\;\;\;\; f'(\alpha) \nonumber \\
  &= \frac{-(p+1)(1-\alpha)^p(\delta-\alpha)^p + (1-\alpha)^{p+1}p(\delta-\alpha)^{p-1}}{(\delta - \alpha)^{2p}} \\
  &= \frac{(\delta-\alpha)^{p-1}(1-\alpha)^p}{(\delta-\alpha)^{2p}} (\alpha - (p+1)\delta + p).
\end{align*}

Set $f'(\alpha) = 0$ and we get $\alpha = (p+1)\delta - p$.

It is ready to calculate that 
\begin{align*}
V(\p_{\alpha}) =  
    \begin{cases}
      \dfrac{\D^p}{2^p (p+1)\delta^p}, & \text{for } \alpha=0 \\
      \dfrac{(p+1)^{p}}{2^p p^p} (1-\delta) \D^p,       & \text{for } \alpha = (p+1)\delta - p
    \end{cases}
\end{align*}

It is easy to see that when $\delta \le \frac{p}{p+1}$, at the point where the derivative is zero we have $\alpha = (p+1)\delta - p < 0$,
and thus the minimum is achieved at $\alpha = 0$. When $\delta > \frac{p}{p+1}$, $(p+1)\delta -p  \in (0, \delta)$, and we have 
$V(\p_{0})  \ge V(\p_{(p+1)\delta-p})$. Indeed,
\begin{align}
&\;\;\;\;\; V(\p_{0})  \ge V(\p_{(p+1)\delta-p})  \nonumber \\
& \Leftrightarrow  \frac{\D^p}{2^p (p+1)\delta^p} \ge \frac{(p+1)^{p}}{2^p p^p} (1-\delta) \D^p \nonumber \\
& \Leftrightarrow  \frac{p^p}{(p+1)^{p+1}} \ge \delta^p(1-\delta), \label{eq:inequiality}
\end{align}
where \eqref{eq:inequiality} holds as 
\begin{align*}
\delta^p(1-\delta) & =  \frac{\delta^p(p-p\delta)}{p} \\
& \le \frac{(\frac{(p\delta + p - p\delta)}{p+1})^{p+1}}{p} \\
& = \frac{p^p}{(p+1)^{p+1}}.
\end{align*}
Therefore, when $\delta > \frac{p}{p+1}$, the minimum of $f(\alpha)$ is achieved at $\alpha = (p+1)\delta - p$.

In conclusion, for the $\ell^p$ cost function, the optimal $\alpha$ is
\begin{align*}
\alpha^*=
  \begin{cases}
      0,           & \text{for }  \delta \in (0, \frac{p}{p+1}]  \\
      (p+1)\delta - p, & \text{for }  \delta \in (\frac{p}{p+1}, 1)
    \end{cases}
\end{align*}
and the minimum cost is
\begin{align*}
V(\p_{\alpha^*}) =  
    \begin{cases}
      \frac{\D^p}{2^p (p+1)\delta^p}, & \text{for }  \delta \in (0, \frac{p}{p+1}] \\
      \frac{(p+1)^{p}}{2^p p^p} (1-\delta) \D^p, & \text{for } \delta \in (\frac{p}{p+1}, 1)
    \end{cases}
\end{align*}
\end{proof}
% end of proof of theorem for general l^n cost function

Applying Theorem \ref{thm:noise-general-n} to the cases where $p=1$ and $p=2$, we derive the optimal noise probability distribution for $(0,\delta)$-differential privacy with minimum noise amplitude and minimum noise power, respectively.

% l^1 cost function, minimum noise amplitude
\begin{corollary}[Optimal Noise Amplitude]\label{cor:noise-amplitude}
Given $0 < \delta <1$ and the query sensitivity $\D > 0$, to minimize the expectation of the amplitude of the noise (i.e., for the $\ell^1$ cost function), the optimal noise probability distribution is $\p_\alpha*$ with
\begin{align*}
\alpha^*=
    \begin{cases}
      0,           & \text{for } \delta \in (0, \frac{1}{2}] \\
      2\delta - 1, & \text{for } \delta \in (\frac{1}{2}, 1)
    \end{cases}
\end{align*}
and the minimum expectation of noise amplitude is
\begin{align}
V(\p_{\alpha^*}) =  
    \begin{cases}
      \frac{\D}{4\delta}, & \text{for } \delta \in (0, \frac{1}{2}] \\
      (1-\delta)\D,       & \text{for } \delta \in (\frac{1}{2}, 1)
    \end{cases}\label{eq:optimal-noise-l1}
\end{align}
\end{corollary}

% l^2 cost function, minimum noise power
\begin{corollary}[Optimal Noise Power]\label{cor:noise-power}
Given $0 < \delta <1$ and the query sensitivity $\D > 0$, to minimize the expectation of the power of the noise (i.e., for the $\ell^2$ cost function), the optimal noise probability distribution is $\p_\alpha*$ with
\begin{align*}
\alpha^*=
    \begin{cases}
      0,           & \text{for } \delta \in (0, \frac{2}{3}] \\
      3\delta - 2, & \text{for } \delta \in (\frac{2}{3}, 1)
    \end{cases}
\end{align*}
and the minimum expectation of noise power is
\begin{align}
V(\p_{\alpha^*}) =  
    \begin{cases}
      \frac{\D^2}{12\delta^2},     & \text{for } \delta \in (0, \frac{2}{3}] \\
      \frac{9}{16}(1-\delta)\D^2 , & \text{for } \delta \in (\frac{2}{3}, 1)
    \end{cases}\label{eq:optimal-noise-l2}
\end{align}
\end{corollary}

\section{Comparison to Gaussian Mechanism} \label{sec:comparison}
In this section, we compare our result with the classic Gaussian mechanism, which adds a Gaussian noise to preserve $(\e, \delta)$-differential privacy. We show a two-fold and three-fold improvement in the high privacy regime over the improved Gaussian mechanism in \cite{icmlGaussian} for minimizing the noise magnitude and the noise power, and we show that the gain is more pronounced in the low privacy regime. 

A classic result on the Gaussian mechanism is that for any $\e, \delta \in (0,1)$, adding a Gaussian noise with standard deviation $\sigma = \frac{\sqrt{2 \log (1.25/\delta)}}{\e}\D$ preserves $(\e, \delta)$-differential privacy \citep{DPbook}. This result does not apply to the $(0, \delta)$-differential privacy, as this would require $\sigma$ to be $+\infty$ when $\e = 0$. 

For $(\e,\delta)$-differential privacy, \cite{icmlGaussian} developed an optimal Gaussian mechanism whose variance is calibrated directly using the Gaussian cumulative density function instead of a tail bound approximation. \cite{icmlGaussian} show the following:
\begin{theorem}[Theorem 2 in \cite{icmlGaussian}]
A Gaussian output perturbation mechanism with $\sigma=\frac{\D}{2\delta}$ preserves $(0, \delta)$-differential privacy.
\end{theorem}

It is ready to see that the Gaussian noise distribution has an expected noise amplitude $\sigma=\frac{\D}{2\delta}$ and an expected noise power $\sigma^2=\frac{\D^2}{4\delta^2}$.
\begin{table}[H]
\centering
\caption{Noise Magnitude Comparison} \label{tab:guassian}
\begin{tabular}{c|c} 
 \hline
 \rule{0pt}{4ex}  Gaussian & $\frac{\D}{2\delta}$ \\
 \hline
 \rule{0pt}{4ex}  Optimal &   $\begin{cases}
    \frac{\D}{4\delta}, & \text{for } \delta \in (0, \frac{1}{2}] \\
    (1-\delta)\D,       & \text{for } \delta \in (\frac{1}{2}, 1)
  \end{cases}$ 
\end{tabular}
\end{table}

\begin{table}[H]
\centering
    \caption{Noise Power Comparison} \label{tab:guassian2}
    \begin{tabular}{c|c} 
     \hline
     \rule{0pt}{4ex}  Gaussian & $\frac{\D^2}{4\delta^2}$ \\
     \hline
     \rule{0pt}{4ex}  Optimal & $   \begin{cases}
        \frac{\D^2}{12\delta^2},     & \text{for } \delta \in (0, \frac{2}{3}] \\
        \frac{9}{16}(1-\delta)\D^2 , & \text{for } \delta \in (\frac{2}{3}, 1)
      \end{cases}$ 
    \end{tabular}
\end{table}

For comparison, our result \eqref{eq:optimal-noise-l1} and \eqref{eq:optimal-noise-l2} in this paper show that the minimum expected noise magnitude and noise power are $\frac{\D}{4\delta}$ and $\frac{\D^2}{12\delta^2}$ in the medium/high privacy regime ($\delta \le \frac{1}{2}$). Therefore, our result shows a two-fold and three-fold multiplicative gain over the improved Gaussian mechanism in \cite{icmlGaussian} for $(0, \delta)$-differential privacy in the high privacy regime for minimizing the noise magnitude and the noise power, respectively. 

In the low privacy regime, the gap is more pronounced: as $\delta \to 1$, the cost of the Gaussian mechanism converges to $\frac{\D}{2}$ and $\frac{\D^2}{4}$ for the noise magnitude and the noise power, while the cost of optimal noise from \eqref{eq:optimal-noise-l1} and \eqref{eq:optimal-noise-l2} converges to zero proportionally to $(1-\delta)$ .

% \begin{table}[H]
% \caption{Cost Comparison of Optimal Mechanism and Improved Gaussian Mechanism} \label{tab:comparison}
% \begin{tabular}{c|c|c} 
%  & Noise Amplitude & Noise Power    \\
%  \hline
%  \rule{0pt}{4ex}  Gaussian & $\frac{\D}{2\delta}$ &  $\frac{\D^2}{4\delta^2}$ \\
%  \rule{0pt}{4ex}  Optimal & $\frac{\D}{4\delta}$ &  $\frac{\D^2}{12\delta^2}$ 
% \end{tabular}
% \end{table}

We plot the ratio of the optimal noise magnitude and noise power over Gaussian mechanism in Fig.~\ref{fig:ratio}. We conclude that the derived optimal $(0,\delta)$-differential private mechanism in this work reduces the noise magnitude and noise power by $\frac{1}{2}$ and $\frac{2}{3}$ in the high privacy regime, and the improvement is more pronounced in the low privacy regime.

\begin{figure}[H]
    \centering
    \includegraphics[width=0.5\textwidth]{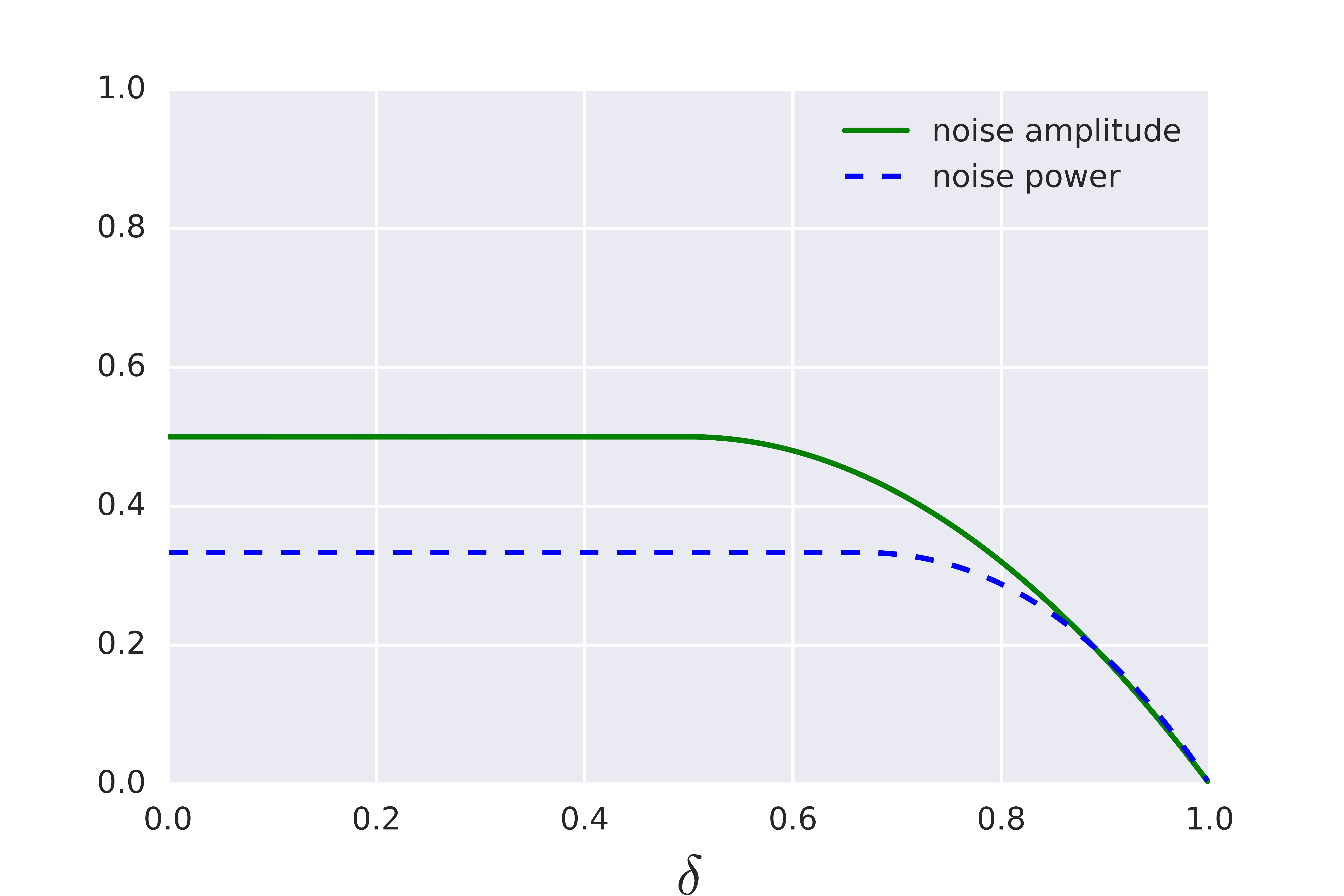}
    \caption{Ratio of the Optimal Noise Cost over the Improved Gaussian Mechanism.}
    \label{fig:ratio}
\end{figure}

\section*{Acknowledgment}\label{sec:acknowledgment}
We would like to thank the anonymous reviewers for their insightful comments and suggestions, which help us improve the presentation of this work.

%%%%%%%%%%%%%%% Reference %%%%%%%%%%%%%%%%%%%%
\bibliographystyle{plainnat}
\bibliography{reference}

\begin{thebibliography}{27}
\providecommand{\natexlab}[1]{#1}
\providecommand{\url}[1]{\texttt{#1}}
\expandafter\ifx\csname urlstyle\endcsname\relax
  \providecommand{\doi}[1]{doi: #1}\else
  \providecommand{\doi}{doi: \begingroup \urlstyle{rm}\Url}\fi

\bibitem[Abadi et~al.(2016)Abadi, Chu, Goodfellow, McMahan, Mironov, Talwar,
  and Zhang]{Abadi2016}
Martin Abadi, Andy Chu, Ian Goodfellow, H.~Brendan McMahan, Ilya Mironov, Kunal
  Talwar, and Li~Zhang.
\newblock Deep learning with differential privacy.
\newblock In \emph{Proceedings of the 2016 ACM SIGSAC Conference on Computer
  and Communications Security}, CCS '16, pages 308--318. ACM, 2016.

\bibitem[Agarwal et~al.(2018)Agarwal, Suresh, Yu, Kumar, and
  McMahan]{Agarwal18}
Naman Agarwal, Ananda~Theertha Suresh, Felix Yu, Sanjiv Kumar, and Brendan
  McMahan.
\newblock cp{SGD}: Communication-efficient and differentially-private
  distributed {SGD}.
\newblock In \emph{Advances in Neural Information Processing Systems}. 2018.

\bibitem[Balle and Wang(2018)]{icmlGaussian}
Borja Balle and Yu-Xiang Wang.
\newblock Improving the {G}aussian mechanism for differential privacy:
  Analytical calibration and optimal denoising.
\newblock In \emph{Proceedings of the 35th International Conference on Machine
  Learning (ICML)}, 2018.

\bibitem[Chaudhuri and Monteleoni(2008)]{CM08}
Kamalika Chaudhuri and Claire Monteleoni.
\newblock {Privacy-preserving logistic regression}.
\newblock In \emph{Neural Information Processing Systems}, pages 289--296,
  2008.

\bibitem[Chaudhuri et~al.(2011)Chaudhuri, Monteleoni, and Sarwate]{ERM}
Kamalika Chaudhuri, Claire Monteleoni, and Anand~D. Sarwate.
\newblock Differentially private empirical risk minimization.
\newblock \emph{Journal of Machine Learning Research}, 12:\penalty0 1069--1109,
  2011.

\bibitem[Chaudhuri et~al.(2012)Chaudhuri, Sarwate, and Sinha]{PCA}
Kamalika Chaudhuri, Anand Sarwate, and Kaushik Sinha.
\newblock Near-optimal differentially private principal components.
\newblock In \emph{Advances in Neural Information Processing Systems 25}, pages
  989--997. 2012.

\bibitem[Duchi et~al.(2012)Duchi, Jordan, and Wainwright]{Duchi12}
John Duchi, Michael Jordan, and Martin Wainwright.
\newblock Privacy aware learning.
\newblock In \emph{Advances in Neural Information Processing Systems}, pages
  1430--1438, 2012.

\bibitem[Dwork(2008)]{DPsurvey}
Cynthia Dwork.
\newblock {Differential Privacy: A Survey of Results}.
\newblock In \emph{Theory and Applications of Models of Computation}, volume
  4978, pages 1--19, 2008.

\bibitem[Dwork and Roth(2014)]{DPbook}
Cynthia Dwork and Aaron Roth.
\newblock The algorithmic foundations of differential privacy.
\newblock \emph{Foundations and Trends in Theoretical Computer Science},
  9\penalty0 (3-4):\penalty0 211--407, 2014.

\bibitem[Dwork et~al.(2006{\natexlab{a}})Dwork, Kenthapadi, McSherry, Mironov,
  and Naor]{DKMMN06}
Cynthia Dwork, Krishnaram Kenthapadi, Frank McSherry, Ilya Mironov, and Moni
  Naor.
\newblock Our data, ourselves: privacy via distributed noise generation.
\newblock In \emph{Proceedings of the 24th annual international conference on
  The Theory and Applications of Cryptographic Techniques}, EUROCRYPT'06, pages
  486--503. Springer-Verlag, 2006{\natexlab{a}}.

\bibitem[Dwork et~al.(2006{\natexlab{b}})Dwork, McSherry, Nissim, and
  Smith]{DMNS06}
Cynthia Dwork, Frank McSherry, Kobbi Nissim, and Adam Smith.
\newblock Calibrating noise to sensitivity in private data analysis.
\newblock In \emph{Theory of Cryptography}, volume 3876 of \emph{Lecture Notes
  in Computer Science}, pages 265--284. Springer Berlin / Heidelberg,
  2006{\natexlab{b}}.

\bibitem[Ge et~al.(2018)Ge, Wang, Wang, and Liu]{PCA2}
Jason Ge, Zhaoran Wang, Mengdi Wang, and Han Liu.
\newblock Minimax-optimal privacy-preserving sparse pca in distributed systems.
\newblock In \emph{Proceedings of the Twenty-First International Conference on
  Artificial Intelligence and Statistics (AISTATS)}, 2018.

\bibitem[Geng and Viswanath(2014)]{GV_ISIT14}
Quan Geng and Pramod Viswanath.
\newblock The optimal mechanism in differential privacy.
\newblock In \emph{IEEE International Symposium on Information Theory (ISIT)},
  pages 2371--2375, June 2014.

\bibitem[Geng and Viswanath(2016{\natexlab{a}})]{GV_IT_Approximate}
Quan Geng and Pramod Viswanath.
\newblock Optimal noise adding mechanisms for approximate differential privacy.
\newblock \emph{IEEE Transactions on Information Theory}, 62\penalty0
  (2):\penalty0 952--969, Feb 2016{\natexlab{a}}.

\bibitem[Geng and Viswanath(2016{\natexlab{b}})]{GV_IT_Epsilon}
Quan Geng and Pramod Viswanath.
\newblock The optimal noise-adding mechanism in differential privacy.
\newblock \emph{IEEE Transactions on Information Theory}, 62\penalty0
  (2):\penalty0 925--951, Feb 2016{\natexlab{b}}.

\bibitem[Geng et~al.(2015)Geng, Kairouz, Oh, and Viswanath]{GV_2_Dimension}
Quan Geng, Peter Kairouz, Sewoong Oh, and Pramod Viswanath.
\newblock The staircase mechanism in differential privacy.
\newblock \emph{IEEE Journal of Selected Topics in Signal Processing},
  9\penalty0 (7):\penalty0 1176--1184, Oct 2015.

\bibitem[Geng et~al.(2018)Geng, Ding, Guo, and Kumar]{GW2_18}
Quan Geng, Wei Ding, Ruiqi Guo, and Sanjiv Kumar.
\newblock {Truncated Laplacian Mechanism for Approximate Differential Privacy}.
\newblock \emph{ArXiv e-prints}, October 2018.

\bibitem[Ghosh et~al.(2009)Ghosh, Roughgarden, and Sundararajan]{Ghosh09}
Arpita Ghosh, Tim Roughgarden, and Mukund Sundararajan.
\newblock Universally utility-maximizing privacy mechanisms.
\newblock In \emph{Proceedings of the 41st annual ACM symposium on Theory of
  computing}, STOC '09, pages 351--360. ACM, 2009.

\bibitem[Gupte and Sundararajan(2010)]{minimax10}
Mangesh Gupte and Mukund Sundararajan.
\newblock {Universally optimal privacy mechanisms for minimax agents}.
\newblock In \emph{Symposium on Principles of Database Systems}, pages
  135--146, 2010.

\bibitem[Jain et~al.(2012)Jain, Kothari, and Thakurta]{Jain12}
Prateek Jain, Pravesh Kothari, and Abhradeep Thakurta.
\newblock Differentially private online learning.
\newblock In \emph{Proceedings of the 25th Annual Conference on Learning Theory
  (COLT)}, 2012.

\bibitem[Jain et~al.(2018)Jain, Thakkar, and Thakurta]{JainMC}
Prateek Jain, Om~Dipakbhai Thakkar, and Abhradeep Thakurta.
\newblock Differentially private matrix completion revisited.
\newblock In \emph{Proceedings of the 35th International Conference on Machine
  Learning (ICML)}, 2018.

\bibitem[Mironov(2017)]{RenyiDP}
Ilya Mironov.
\newblock R\'{e}nyi differential privacy.
\newblock In \emph{2017 IEEE 30th Computer Security Foundations Symposium
  (CSF)}, pages 263--275, Aug. 2017.

\bibitem[Park et~al.(2017)Park, Foulds, Chaudhuri, and Welling]{EM}
Mijung Park, James Foulds, Kamalika Chaudhuri, and Max Welling.
\newblock {DP-EM: Differentially Private Expectation Maximization}.
\newblock In \emph{Proceedings of the 20th International Conference on
  Artificial Intelligence and Statistics (AISTATS)}, 2017.

\bibitem[Phan et~al.(2016)Phan, Wang, Wu, and Dou]{Phan2016}
Ngoc-Son Phan, Yue Wang, Xintao Wu, and Dejing Dou.
\newblock Differential privacy preservation for deep auto-encoders: an
  application of human behavior prediction.
\newblock In \emph{AAAI}, 2016.

\bibitem[Sheffet(2018)]{HT18}
Or~Sheffet.
\newblock Locally private hypothesis testing.
\newblock In \emph{Proceedings of the 35th International Conference on Machine
  Learning (ICML)}, 2018.

\bibitem[Shokri and Shmatikov(2015)]{Shokri15}
Reza Shokri and Vitaly Shmatikov.
\newblock Privacy-preserving deep learning.
\newblock In \emph{Proceedings of the 22Nd ACM SIGSAC Conference on Computer
  and Communications Security}, CCS '15, pages 1310--1321. ACM, 2015.

\bibitem[Soria-Comas and Domingo-Ferrer(2013)]{DomingoFerrer2013}
Jordi Soria-Comas and Josep Domingo-Ferrer.
\newblock Optimal data-independent noise for differential privacy.
\newblock \emph{Information Sciences}, 250:\penalty0 200 -- 214, 2013.

\end{thebibliography}

\end{document}